\documentclass[letterpaper,12pt,reqno]{amsart}
\usepackage[margin=1.2in]{geometry}
\usepackage{eucal}
\usepackage{tikz}
\usepackage{bbm}
\usepackage{etex}
\usepackage{graphicx}
\usepackage{algorithm}
\usepackage[noend]{algpseudocode}
\usepackage{stackrel}
\usepackage{sseq}
\usepackage{tikz-cd}
\usepackage{amsmath,amssymb}               
\usepackage{verbatim}
\usepackage{xypic}
\usepackage{amsfonts}
\usepackage{epsfig}
\usepackage{enumerate}
\usepackage{enumitem}
\usepackage{ mathrsfs }
\usepackage{amsthm}
\usepackage{mathtools}
\usetikzlibrary{decorations.pathreplacing,calligraphy,arrows}
\usepackage{baskervald}
\usepackage[breaklinks,colorlinks,citecolor=teal,linkcolor=teal,urlcolor=teal,pagebackref,hyperindex]{hyperref}
\usepackage{lipsum}

\newcommand\blfootnote[1]{%
	\begingroup
	\renewcommand\thefootnote{}\footnote{#1}%
	\addtocounter{footnote}{-1}%
	\endgroup
}

\newtheorem{theorem}{Theorem}[section]
\newtheorem{lemma}[theorem]{Lemma}

\newtheorem*{question*}{Question}
\newtheorem{prop}[theorem]{Proposition}
\newtheorem{corollary}[theorem]{Corollary}

\theoremstyle{definition}
\newtheorem{defn}[theorem]{Definition}
\newtheorem{remark}[theorem]{Remark}

\newtheorem*{notation*}{Notation}

\newcommand{\Var}{\mathrm{Var}}
\newcommand{\p}{{\partial}}

\newcommand{\K}{\mathbf{K}}
\newcommand{\bl}{\mathbf{L}}
\newcommand{\E}{\mathbb{E}}
\newcommand{\bp}{\mathbb{P}}
\newcommand{\fp}{\mathfrak{P}}
\newcommand{\hk}{\hat{K}}
\newcommand{\hl}{\hat{L}}

\numberwithin{equation}{section}

\title{Continuous Blackjack: Equilibrium, Deviation \& Adaptive Strategy}

\author{MU ZHAO}

\begin{document}

	\begin{abstract}
	 We introduce a variant of the classic poker game blackjack --- the continuous blackjack. We study the Nash Equilibrium as well as the cases where players deviate from it. We then pivot to the study of a large class of adaptive strategies and obtain a model free strategy. Finally, we apply reinforcement learning techniques to the game and address several associated engineering challenges.
	\end{abstract}
	\blfootnote{e-mail: \emph{muzhao.pku@gmail.com}}
	\blfootnote{Python code for Continuous Blackjack can be found \href{https://github.com/mu-zhao/Continuous-Blackjack}{here}.}
	\maketitle
	\tableofcontents
	\section{Introduction}
	 In~\cite{10.2307/27642600},
	 S. Herschkorn introduced the continuous Blackjack game --- a variation of the 
	 well-known poker game Blackjack and computed its Nash Equilibrium. In that game, 
	 each player can take one or two numbers from standard uniform distribution as opposed 
	 to the cards. Similarly, we introduce a new variant of the classic Blackjack game, with no limitation on how many hits one player can have. The following are the detailed rules:
	\begin{itemize}
		\item The game will run for many rounds, say, one billion.
		\item Each round every player's position is reshuffled randomly.
		\item From the first player, each player plays their turn in order, and other players can observe the previous players' actions.
		\item Each player can choose to hit or stay. If the player choose to hit, a random number is generated from standard uniform distribution and added to the player's total sum; otherwise the player's turn ends.
		\item In each round, players' scores are the total sums as long as they don't exceed $1$, in which case the score will be $0$. At the end of each round, the player with the highest score receive one point.
		\item  In the rare scenario where two or more players get the same highest score, they 
		 will share the point equally among them.
	\end{itemize}

\begin{question*}
What's the optimal strategy to maximize the long term payoff?
\end{question*}
This paper answer the question in the following way:
first, we find the Nash Equilibrium in chapter~\ref{sec:ne}; 
secondly, we study the case where players deviate from Nash Equilibrium in chapter~\ref{sec:deviation}. In chapter~\ref{sec:adaptive}, we pivot to the study of a large class of adaptive strategies: under weak assumptions, we obtain important results including several \textbf{upper bounds} and come up with a \textbf{model free strategy}. Finally, in chapter~\ref{sec:RL}, we model the game on \textbf{contextual bandits} and apply reinforcement learning technique to it.

\section{Nash Equilibrium}\label{sec:ne}
To compute Nash Equilibrium, we need the following lemma:
\begin{lemma}
Let $\{X_i\}$ be i.i.d uniform distributions on $[0,1]$. For any $t\in [0,1]$, let $N(t)$ be the stopping time:
\[ N(t):=\min\{ n\,|\, X_1+\cdots+X_n>t \},
\]
and 
\[ S(t):=\sum_{i=1}^{N(t)}X_i.
\]
 For $\,0\leq x \leq y \leq 1$, let $F(x,y):=\bp(S(x)\leq y)$, then we have 
\begin{equation}
 F(x,y)=(y-x)e^x.
\end{equation}
\end{lemma}
\begin{proof}
We will calculate the probability $P(S(x)\leq y)$ conditioned on $X_1$.
\begin{align*}
F(x,y)& =\bp(S(x)\leq y, X_1\leq x)+P(S(x)\leq y, x<X_1\leq y)\\
 &=\int_0^x \bp(S(x)\leq y|X_1=t)dt+(y-x)\\
 &=\int_0^x \bp(S(x-t)\leq y-t) dt +(y-x)\\
 &=\int_0^x F(x-t,y-t)dt +(y-x)
\end{align*}
Let $z=y-x$, we have
\[ F(x,x+z)=z+\int_0^x F(x-t,x+z-t)dt=z+\int_0^x F(t,z+t)dt
\]
Let $g(x,z)=F(x,x+z)$, then we have
\[ g(x,z)=z+\int_0^x g(t,z)dt
\]
Notice that $g(x,0)=0$, we have $g(x,z)=ze^z$
\end{proof}	
Now we can proceed to find the Nash Equilibrium. Let's find out the best choice of the first player for the $k+1$-player game. We denote $\alpha_k$ the best choice of the first player among $k+1$ players if everyone plays optimally. Clearly, we have $\alpha_0=0$.
\subsection{Simple Case}
Now let's compute $\alpha_1$.
Suppose for a two-player game, the threshold for the first player is $A$ and payoff is $E(A)$. The second player's best response will simply be a threshold exactlly equals the score of the first player. The first player get a point if and only if the second player  go bust. If first players' score is $t$,  then the second player has probability $1-F(t,1)$ to go bust. And therefore we have:
\begin{align*}
E(A)&=\int_A^1[1-F(t,1)]\frac{\p F}{\p y}(A,t)dt\\
&=e^A\int_A^1[1-F(t,1)]dt.
\end{align*}
The derivative is
\begin{align*}
\frac{\p E}{\p A}(A,1)=e^A\bigg[\int_A^1[1-F(t,1)]dt-[1-F(A,1)]\bigg].
\end{align*}
Therefore, we have $\alpha_1$ satisfies
\[ 1-F(\alpha_1,1)=\int_{\alpha_1}^1[1-F(t,1)]dt.
\]
We can interpret the above equation as follows: \textbf{at the  moment when the score is $A$, 
the left-hand side represents the payoff if we stop there, the right-hand side represents the 
payoff if we hit exactly one more time. This is the critical threshold where  player is indifference 
to whether stopping there or  hitting one more time. It is also worth mentioning that the left-hand
 side is increasing while the right-hand side is decreasing.}
\begin{remark}
 $E(A)$ is increasing on the left of $\alpha_1$ while decreasing on the right of $\alpha_1$.
\end{remark}
\subsection{General Case}
Now let's look at the general case.
\begin{theorem}
Suppose $n+1$ player play this game optimally, the best strategy for the first player is a simple threshold $\alpha_n$ defined by
\begin{equation}\label{eq:def NE}
[1-F(\alpha_n,1)]^n=\int_{\alpha_n}^1[1-F(t,1)]^ndt
\end{equation}
\end{theorem}
\begin{proof}

We will prove by induction. Suppose we already have found  the first $k-1$ thresholds
$\alpha_1,\alpha_2,\cdots,\alpha_{k-1}$ and the payoffs $E_1(x),E_2(x),\cdots,E_{k-1}(x)$, and they satisfy the following conditions:
\begin{itemize}
	\item $\alpha_i<\alpha_j,0<i<j< k$
	\item $E_i'(x)>0,x<\alpha_i$ and $E_i'(x)<0,x>\alpha_i$ for $0<i< k$
	\item $\alpha_i$ is the root of
	\begin{equation}
	[1-F(x,1)]^i=\int_x^1[1-F(t,1)]^idt
	\end{equation}

\end{itemize}
Suppose the first player has sum $t$, and the only way the first player can get point is that everyone else go bust(since they all play optimally and will not stop before reach $t$). Therefore for everyone in this scenario, their best response is $A_i(t)=\max(\alpha_i,t)$.
Therefore, we have
\[ E_k(A)=\int_A^1\prod_{i=1}^k[1-F(A_i(t),1)]\frac{\p F}{\p y}(A,t)dt
\]
and 
\[
\frac{dE_k(A)}{dA}=e^A\Bigg[\int_A^1\prod_{i=1}^k[1-F(A_i(t),1)]dt-\prod_{i=1}^k[1-F(A_i(A),1)]\Bigg]
\]
Notice \[h(A):=\Bigg[\int_A^1\prod_{i=1}^k[1-F(A_i(t),1)]dt-\prod_{i=1}^k[1-F(A_i(A),1)]\Bigg]
\]
and $\frac{dE_k(A)}{dA}$ have the same zeros. $h(A)$ is decreasing, and 
\[h(1)=-1,\] and we have $1-F(A_i(t),1)\geq 1-F(A_i,1)$,therefore 
\[h(0)=\int_0^1\prod_{i=1}^k[1-F(A_i(t),1)]dt-\prod_{i=1}^k[1-F(A_i,1)]>0
\]
Now we conclude that there is a unique $\alpha_k$. Furthermore, $E_k'(A)>0$ for $A<\alpha_k$ and $E_k'(A)<0$ for $A>\alpha_k$. Now we only need to prove $\alpha_k>\alpha_{k-1}$.

We prove by contradiction: suppose $\alpha_{j-1}\leq \alpha_k<\alpha_j$, then we have
\[
\int_{\alpha_k}^1\prod_{i=1}^j[1-F(t,1)]\prod_{i=j}^{k-1}[1-F(\max(\alpha_i,t),1)]dt=\prod_{i=1}^j[1-F(\alpha_k,1)]\prod_{i=j}^{k-1}[1-F(\alpha_i,1)]
\]
since we have
\[ \prod_{i=j}^{k-1}[1-F(\max(\alpha_i,t),1)] > \prod_{i=j}^{k-1}[1-F(\alpha_i,1)]\,, \quad t>\alpha_k
\]
which gives rise to
\[
\int_{\alpha_k}^1[1-F(t,1)]^jdt<[1-F(\alpha_k,1)]^j
\]
Therefore,
\[
\int_{\alpha_j}^1[1-F(t,1)]^jdt< \int_{\alpha_k}^1[1-F(t,1)]^jdt<[1-F(\alpha_k,1)]^j
< [1-F(\alpha_j,1)]^j
\]
Meanwhile, by the definition of $\alpha_j$, we have
\[
\int_{\alpha_j}^1[1-F(t,1)]^jdt=[1-F(\alpha_j,1)]^j
\]
Contradiction!
\end{proof}
Here's the table for the first few terms for the Nash Equilibrium.
\begin{table}[ht]
	\caption{Nash Equilibrium $\alpha_n$} 
	\centering 
	\begin{tabular}{c c c c c c c c } 
		\hline\hline
		 n & 1&2&3&4&5&6&7  \\ [0.5ex] 
		\hline 
		
		$\alpha_n$& 0.570557& 0.687916& 0.748671& 0.787111&
		0.814059& 0.834191 & 0.849900\\ [0.7ex] 
		\hline 
		n &8& 9&10&11&12&13&14  \\ [0.5ex] 
		\hline 
		
	$\alpha_n$& 0.862558& 0.873008&
		0.881805& 0.889328 &0.895845&0.901554 &0.906602\\[0.7ex]
		\hline\hline 
	\end{tabular}
	\label{table:nonlin} 
\end{table}

\section{Away from Nash Equilibrium}\label{sec:deviation}
Now that we have computed the Nash Equilibrium, the next step is to study the case where players are not playing the
 Nash Equilibrium. In theory, players can play any strategy. For all practical purposes, we make the assumption that the player's strategy will be thresholds.
 
\subsection{Pure Strategy}
 
 Suppose the $n+1$ players are playing the strategies of 
  the simple thresholds $A,k_1,k_2,\cdots,k_n$, and let 
\begin{equation}\label{eq: def G}
G(t,k) = F(k,\max(t,k))+1-F(k,1)= \begin{cases}
1-F(k,1), \quad t<k \\
1+F(k,t)-F(k,1),\quad t\geq k
\end{cases}
\end{equation}
Then the expected payoff of the first player is 
\begin{equation}\label{eq:payoff}
E(A,k_1,\cdots,k_n)=e^A\int_A^1 H(t,k_1,\cdots,k_n)dt
\end{equation}
where 
\begin{equation}\label{eq: def H}
H(t,k_1,\cdots,k_n)=G(t,k_1)G(t,k_2)\cdots G(t,k_n)
\end{equation}
As in section~\ref{sec:ne}, the optimal strategy for the first player is to have 
threshold $A(k_1,k_2,\cdots,k_n)$ satisfy
\begin{equation}\label{eq:optimal A}
\int_A^1 H(t)dt = H(A)
\end{equation}
where we have suppressed $k_i's$ in the notations $H$ and $A$.

We are interested in the direction of changes of the optimal response $A$ corresponding 
to the shift in $k_i$, that is, we are interested in the partial 
derivatives $\frac{\p A}{\p k_i}$. However, the function $H$ is not smooth. 
To deal with this issue, we instead will investigate the left and right derivatives. Denote
\[\p_i^\pm:=\frac{\p ^\pm}{\p k_i},\quad \p_t^\pm:=\frac{\p ^\pm}{\p t},
\]

We have 
\begin{equation*}
\p_t^+G(t,k) = \begin{cases}
0, \quad t<k \\
e^k,\quad t\geq k
\end{cases}
\end{equation*}
and
\begin{equation*}
\p_k^+G(t,k) = \begin{cases}
ke^k, \quad t< k \\
-(1-t)e^k,\quad t\geq k
\end{cases}
\end{equation*}
$\p_*^-G$ coincides with $\p_*^+G$ except for $t=k$.
Take logarithm of equation~\ref{eq: def H} and we get
\[ \ln H = \sum \ln G(t,k_i).
\]
Differentiate both sides, we get
\begin{equation}
\frac{\p_t^+ H}{H}=\sum \frac{\p_t^+ G(t,k_i)}{G(t,k_i)}=\sum_{k_i\leq t} \frac{e^{k_i}}{G(t,k_i)}
\end{equation}
and 
\begin{equation}\label{eq: pH}
\frac{\p_i^+ H}{H}=\frac{\p_i^+ G(t,k_i)}{G(t,k_i)}
\end{equation}
Now we differentiate equation~\ref{eq:optimal A} and get
\begin{equation}
\int_A^1 \p_i^+ H(t) dt - H(A)\p_i^+A=\p_i^+H(A)+\p_t^+H(A)\p_i^+ A.
\end{equation}
That is equivalent to(dividing by $H(A)$)
\[
\frac{1}{H(A)}\int_A^1 \p_i^+ H(t) dt -\p_i^+A=\frac{\p_i^+H(A)}{H(A)}+\frac{\p_t^+H(A)}{H(A)}\p_i^+ A,
\]
i.e,
\begin{align}
\bigg(\frac{\p_t^+H(A)}{H(A)}+1\bigg)\p_i^+ A &=
\frac{1}{H(A)}\int_A^1 \p_i^+ H(t) dt-\frac{\p_i^+H(A)}{H(A)}\\
&=\int_A^1 \frac{H(t)}{H(A)}\frac{\p_i^+H(t)}{H(t)}dt-\frac{\p_i^+H(A)}{H(A)}. \label{eq: pos neg}
\end{align}
By equation~\ref{eq: pH}, we have
\[ \frac{\p_i^+H(t)}{H(t)}=\frac{\p_i^+ G(t,k_i)}{G(t,k_i)}=\frac{k_ie^{k_i}}{1-(1-k_i)e^{k_i}}>0,\quad t<k_i
\]
and 
\[ \frac{\p_i^+H(t)}{H(t)}=\frac{\p_i^+ G(t,k_i)}{G(t,k_i)}=\frac{-(1-t)e^{k_i}}{1-(1-t)e^{k_i}}<0,\quad t\geq k_i
\]
It is easy to check that $\frac{\p_i^+H(t)}{H(t)}$ is increasing on $[k_i,1]$, which leads to the following facts:
\begin{itemize}
\item $\frac{\p_i^+H(A)}{H(A)}$ is maximum if $A<k_i,$
\item $\frac{\p_i^+H(A)}{H(A)}$ is minimum if $A\geq k_i$.
\end{itemize}
Meanwhile, we notice that
\[ \int_A^1 \frac{H(t)}{H(A)}dt=1
\]
from equation~\ref{eq:optimal A}, with $\frac{H(t)}{H(A)}>0$.
Hence for equation~\ref{eq: pos neg}, we have
\begin{equation*}
\int_A^1 \frac{H(t)}{H(A)}\frac{\p_i^+H(t)}{H(t)}dt-\frac{\p_i^+H(A)}{H(A)}=\begin{cases}
<0,\, A<k_i\\
>0,\, A\geq k_i
\end{cases}
\end{equation*}
On the other hand, 
\begin{equation}\label{eq: coef}
\frac{\p_t^+H(A)}{H(A)}+1=\sum_{k_i\leq A} \frac{e^{k_i}}{G(A,k_i)}+1>0.
\end{equation}
We have the following conclusion:
\begin{prop}\label{prop:derivative}
 The sign of $\p_i^\pm A$ is determined by the following equations:
\begin{align}\label{eq: p A sign}
\textbf{sign}(\p_i^+ A)&=h(A-k_i)\\
\textbf{sign}
(\p_i^- A)&=-h(k_i-A)
\end{align}
where $h(x)$ is defined by
\begin{equation*}
h(x)=\begin{cases}
1,\,x\geq 0\\
-1,\,x<0
\end{cases}
\end{equation*}
\end{prop}
Proposition~\ref{prop:derivative} has a surprising implication:
 the only critical point is \[A=k_i,\,1\leq i\leq n,\]
which leads to the following upper bounds:
\begin{corollary}
The maximum of $A(k_1,\cdots,k_n)$ is determined by \textbf{simple threshold upper bound} $\beta_n$, which is defined by 
\begin{equation}
\int_{\beta_n}^1 H(t,\beta_n,\cdots,\beta_n)dt=H(\beta_n,\beta_n,\cdots,\beta_n)
\end{equation}
That is,
\begin{equation}\label{eq: max A}
1-[1-(1-\beta_n)e^{\beta_n}]^{n+1}=(n+1)e^{\beta_n}[1-(1-\beta_n)e^{\beta_n}]^n
\end{equation}
\end{corollary}
\begin{proof}
Notice that $H(t,A,\cdots,A)=[1-(1-t)e^A]^n$ for $t\geq A$, the integral yields equation~\ref{eq: max A}.
\end{proof}
 We have the first few \emph{simple threshold upper bounds} $\beta_n$ as follows:
\begin{table}[ht]
	\caption{\emph{Simple Threshold Upper Bounds} $\beta_n$} 
	\centering 
	\begin{tabular}{c c c c c c c c } 
		\hline\hline
		n & 1&2&3&4&5&6& 7  \\ [0.5ex] 
		\hline 
		
		$\beta_n$ &0.588650& 0.698942& 0.756234& 0.792694& 0.818387& 0.837665 & 0.852764 \\ [0.7ex] 
		\hline 
		n &8&9&10&11&12&13&14  \\ [0.5ex] 
		\hline 
		
		$\beta_n$& 0.864966& 0.875068& 0.883591& 0.890894&0.897231& 0.902791& 0.907714\\[0.7ex]
		\hline\hline 
	\end{tabular}
	\label{table:max threshold} 
\end{table}

\subsection{Mixed Strategy}

Now we can consider the more general case: if other players randomize their strategies, that is,  $k_i$ is no longer a function but a distribution.
Suppose the other players use mixed strategies from a strategy set $\Omega\subset [0,1]^n$ of pure strategies, and let's denote $K:=(k_1,k_2,\cdots,k_n) \in [0,1]$ to be a pure strategy in $\Omega$ , and $\mu(K)$ its distribution function. We have the payoff
\begin{equation}
E(A,\Omega,\mu)= \int_A^1 \int_\Omega H(t,K)e^A d\mu dt=\int_\Omega \int_A^1 H(t,K)e^Adt d\mu
\end{equation}
Therefore, the best response $A(\omega,\mu)$ is determined by
\begin{equation}
\int_\Omega \Big[ \int_A^1 H(t,K)dt- H(A,K)\Big]d\mu=0
\end{equation}
or
\begin{equation}\label{eq: general op}
\int_A^1\int_\Omega H(t,K)d\mu dt =\int_\Omega H(A,k)d\mu
\end{equation}
\begin{remark}
The above equation~\ref{eq: general op} has one unique solution. The best response is always \textbf{pure}!
\end{remark}
\begin{remark}
	We should mention that the strategies considered here is not adaptive. That is, previous results of the same round is not taken into consideration.
\end{remark}
\begin{theorem}\label{thm:main}
Suppose for any pure strategy $K\in \Omega$, we have the best response $A(K) \in [\alpha,\beta]$. Then for any mixed strategy $(\Omega,\mu)$, the best response $A(\Omega,\mu) \in  [\alpha,\beta]$.
\end{theorem}
\begin{proof}
The proof is rather simple. Notice $\int_A^1 H(t)dt - H(A)$ is decreasing as a function of $A$. For any $A<\alpha$, we have:
\[ \int_A^1 H(t,K)dt- H(A,K)<\int_\alpha^1 H(t,K)dt- H(\alpha,K)\leq 0
\]
and 
\[\int_\Omega \Big[ \int_A^1 H(t,K)dt- H(A,K)\Big]d\mu<0
\]
Therefore  $A(\Omega,\mu)\geq \alpha$. We can prove the right side likewise.
\end{proof}
\begin{remark}
Theorem~\ref{thm:main} shows that the best response of the mixed strategy lies in 
the convex hull of the best response of the pure strategies. And strategies outside
 the convex hull are strictly dominated.
\end{remark}
\begin{corollary}
First player's best response has upper bound defined by equation~\ref{eq: max A}, 
however other players mix their strategies.
\end{corollary}

\section{Adaptive Strategy} \label{sec:adaptive}
Now we are going to consider strategies that take other players' actions into account.  Suppose $n+1$ players are playing the game, and the $i$-th player's strategy is a simple threshold determined by the maximum of all the valid scores of the previous players. Formally, let $\{X_i^j\}$ be independent uniform distributions, we define
\begin{align}
N_0&:= \min\{ m \,| \, X_1^0+X_2^0+\cdots+X_m^0+\cdots>A \},\\
N_i&:=  \min\{ m \,| \, X_1^i+X_2^0+\cdots+X_m^i+\cdots>K_i(T_i))\},\quad i\geq 1,\\
S_i(K_i)&: =\sum_{j=1}^{N_i}X_j^i
\end{align}
where $K_i$ is the function defining the strategy while $T_i$ is defined as follows:
\begin{equation*}
W_i(K_i):=\begin{cases}
S_i,\quad S_i\leq 1\\
0, \quad  S_i>1
\end{cases}
\end{equation*}
and $T_i:=\max_{0\leq j < i}\{ W_j\}$.
$N_i$  is the stopping time for player $i$ , and $S_i$ is the cumulative sum, keep in mind that it goes bust once the sum exceeds $1$, and $W_i$ is the score. 
 $T_i$ is the maximum of all the scores before player $i$. The threshold is $K_i(T_i)$ is a random variable. By abuse of notation, we will denote the cumulative distribution function by $K_i(t,s):=\bp(K_i(t)\leq s)$. 

\subsection{General Model}
Let's denote $\K:=(K_1,\cdots,K_n)$, under the assumptions above, the payoff of the first player is
\begin{align}
E(A,\K)&=\int_A^1\int_0^1\cdots\int_0^1 H(t,K)dK_n(t)\cdots dK_1(t)e^Adt\\
&=e^A\int_A^1 \Big(\prod_{i=1}^n \int_0^1 G(t,s)dK_i(t)(s) \Big)dt
\end{align}
Let's define 
\begin{equation}
L(t,K):=\int_0^1 G(t,s)dK(t)(s)=1-\int_0^1 K(t,s)dG(t)(s)
\end{equation}
$L(t,K)$ has a meaningful interpretation: $L$ is the probability to score less that $t$ when playing strategy $K$. More importantly, $L$ is \textbf{observable}.
Let's denote $\bl(\K):=(L_1(K_1),\cdots,L_n(K_n))$, we define
\begin{equation}
 M(t,\bl):=\prod_{i=1}^n L(t).
\end{equation}
Likewise, $M(t,\bl(\K))$ is the probability to score under $t$ for each player if their strategies are determined by $\K$. That is,
\begin{equation}
L(t,K)=\bp(W(K)\leq t),\quad M(t,\bl(\K))=\prod_{i=1}^n \bp(W(K_i)\leq t)
\end{equation}  We can rewrite the expected payoff as 
\begin{equation}\label{eq: general payoff}
E(A,\K)=e^A\int_A^1 M(t,\K)dt
\end{equation}
and any optimal threshold should satisfy
\begin{equation}\label{eq: general NE}
\int_A^1 M(t,\K)dt=M(A,\K)
\end{equation}

Meanwhile, we have
\begin{align*}
|M(t,\K)-M(t,\tilde{\K})|&=\Big|\prod_{i=1}^nL(t,K_i)-\prod_{i=1}^nL(t,\tilde{K_i})\Big|\\
&\leq \sum_{i=1}^n|L(t,K_i)-L(t,\tilde{K_i})|
\end{align*}
Let's define the distance between two strategies as follows:
\begin{equation}
||\bl-\tilde{\bl}||_{L_1}:=\sum_{i=1}^n||L_i-\tilde{L_i}||_{L_1}
\end{equation}
\begin{equation}
d(\K,\tilde{\K}):=||\bl(\K)-\bl(\tilde{\K})||_{L_1}
\end{equation}
then we have
\begin{equation}
|E(A,\bl)-E(A,\tilde{\bl})|\leq e^A(1-A)||\bl-\tilde{\bl}||_{L_1}\leq ||\bl-\tilde{\bl}||_{L_1}= d(\K,\tilde{\K})
\end{equation}
This ensures us that \textbf{ if our strategies are close enough (under $L_1-$norm), then the expected payoff won't be far off.}

We can maximize our payoff as long as we know $L_i$. However, that's generally not the case. Therefore, we have to estimate $L_i$, which is possible if the game is played repetitively for lots of rounds and the other players' strategies are stationary.  
Here's what we are going to do: build for every player a profile. Suppose other players' strategies are determined by
their own position and the maximum of the previous valid scores only, that is, by the tuple of distributions $(L^{(0)},L^{(1)},\cdots,L^{(n)})$. A \emph{profile} $\fp$ for player is an estimate of these distributions:
\begin{equation}
\mathfrak{P}(\K):=(\hat{L}^{(0)},\hat{L}^{(1)},\hat{L}^{(2)},\cdots,\hl^{(n)})
\end{equation}
Where $\hl^{(i)}$ is the estimate of $L^{(i)}$. 

How do we estimate $L$ then? We define
\begin{equation}
\xi(t,K):=\E I_{\{W(K)<t\}},
\end{equation}
 we then have
\begin{equation}
\E\xi=L(t,K),\quad \Var(\xi)=L(t,K)(1-L(t,K))\leq 1/4
\end{equation}
Therefore, let 
\[\eta(t,K):=\frac{\xi_1(t,K)+\cdots+\xi_N(t,K)}{N},
\]
we have \[
\E\eta=L(t,K),\quad \Var(\eta)\leq \frac{1}{4N}
\]

Let $\hl(t,x,K):=\bp(\eta(t,K)\leq x)$, we have 
\begin{equation}\label{ineq:L1}
\E|\bl(\K)-\fp(\K)|^2\leq \frac{n}{4N}
\end{equation}
 and by Chebyshev's inequality,
\begin{equation}\label{ineq:chebe}
 \bp(|\bl(\K,t)-\fp(\K,t)|>\epsilon)\leq\frac{n}{4N\epsilon^2}
 \end{equation}
 \begin{remark}
 $\fp$ converges to $\bl$ relatively slow (square root). However, the probability of large error diminishes linearly.
 \end{remark}

We notice that $\fp[j]$ is a function. In general, it's not feasible to record all the values. Therefore, we have to discretize them. In order to have control over $t$,  let's assume  $L$ to be Lipschitz continuous.

Suppose we record $\fp$ with step size $1/m$, and let $\fp^\Delta$ be the discretization of $\fp$, that is,
$\Delta:=(\Delta_1,\cdots,\Delta_n)$, where $\Delta_i:=[x_{i-1},x_i),\,x_i=i/m$,
\[
\hl^{\Delta_i}(t):=m\int_{\Delta_i}\hl(t)dt,\,\forall t\in \Delta_i.
\]
\[ \fp^\Delta:=(\hl^\Delta_1,\cdots,\hl_n^\Delta).
\]
We have 
\begin{align*}
	\Big|\int_\Delta \big[\hl_1(t)\cdots\hl_n(t)-\hl_1^\Delta\cdots\hl_n^\Delta\big]dt\Big|
	\leq \sum_{i=1}^n \int_\Delta |\hl_i(t)-\hl_i^\Delta(t)|dt
	\leq n\cdot\frac{c}{4m^2}
\end{align*}
The first inequality holds because of the following fact: for real numbers $a_i,b_i$ such that $|a_i|\leq 1,|b_i|\leq 1$, we have
 \begin{align*}
 |a_1a_2\cdots a_n-b_1\cdots b_n|&\leq |a_1a_2\cdots a_n-b_1 a_2\cdots a_n|+|b_1 a_2\cdots a_n-b_1 b_2 a_3\cdots a_n|\\
& +
 \cdots +|b_1b_2\cdots a_n- b_1\cdots b_n|\\
 &\leq |a_1-b_1|+|a_2-b_2|+\cdots+|a_n-b_n|.
 \end{align*}
 The second inequality is a corollary of the following lemma:

\begin{lemma}\label{lemma:good one}
	Suppose $f(t)$ is a Lipschitz function on $[0,1]$ such that  $|f(x)-f(y)|\leq |x-y|$ and
	$\int_0^1 f(t)dt=0$, then for any $p>0$,\[\int_0^1|f(t)|^pdt\leq \frac{1}{(p+1)2^p}.
	\]
\end{lemma}
\begin{proof}
	First, let's define $I^+:=\{x\,|\,f(x)> 0\},\,I^-:=\{x\,|\,f(x)<0\},$ both of them are open sets since $f(t)$ is continuous. $f_+:=\max(f,0),\,f_-:=\max(-f,0)$, we have
	\[ \int_{I^+}f_+(t)dt=\int_{I^-}f_-(t)dt,\quad
	\int_0^1|f(t)|^pdt=\int_{I+}f_+^p(t)dt+\int_{I^-}f_-^p(t)dt.
	\]
	Now we prove the following inequality: \[\int_{I^+}f_+^p(t)dt\leq \frac{1}{p+1}|I^+|^{p+1}.
	\]
	Indeed, if $I^+$ a union of finite disjoint open intervals $\bigcup [x_k,y_k]$, then
	$f(x_k)=0$ or $f(y_k)=0$. (Because $\int_0^1 f(t)dt=0$.) Without loss of generality, let's assume $f(x_k)=0$, then for any $t\in [x_k,y_k]$,
	\[ f(t)\leq t-x_k, 
	\]
	and therefore \[\int_{x_k}^{y_k}f_+^p(t)dt\leq \frac{(y_k-x_k)^{p+1}}{p+1}.\]
	We sum over the union, and get 
	\[ \int_{I^+}f_+^p(t)dt\leq \frac{1}{p+1}\sum (y_k-x_k)^{p+1}\leq \frac{1}{p+1} |I^+|^{p+1}
	\]
	Since $I^+$ is open, it can be approximated by finite union of disjoint open sets. Therefore, we have the above inequality holds, which gives rise to
	\[ \int_0^1|f(t)|^pdt\leq 2\cdot\min\Big( \int_{I^+}f_+^p(t)dt,\int_{I^-}f_-^p(t)dt\Big)
	\leq \frac{2}{p+1}\min\Big( |I^+|^{p+1},|I^-|^{p+1} \Big)\leq \frac{1}{(p+1)\cdot2^p}.
	\]
\end{proof}

Therefore,

\begin{equation}
	|E(A,\fp)-E(A,\fp^\Delta)|\leq m\cdot \frac{nc}{4m^2}=\frac{nc}{4m}
\end{equation}

Let \[A^*:= \arg\max\limits_{A} E(A,\fp),\quad A^\Delta:= \arg\max\limits_{A} E(A,\fp^\Delta)\] we have
\begin{align*}
	E(A^*,\fp)&\leq E(A^*,\fp^\Delta)+\frac{nc}{4m}\\
	&\leq E(A^\Delta,\fp^\Delta)+\frac{nc}{4m}\\
	&\leq E(A^\Delta,\fp)+\frac{nc}{4m}+\frac{nc}{4m}\\
	&\leq E(A^\Delta,\fp)+\frac{nc}{2m}
\end{align*}
 We can rewrite the inequality as:
\begin{equation}
	E(A^\Delta,\fp)\geq \max(E(A,\fp))-\frac{nc}{2m}
\end{equation}

\begin{remark}
The accuracy of discretization is $O(1/m)$.
\end{remark}
\subsubsection{\textbf{Nonstationary Scenario}}
	Even though we made the assumption that other players' strategies are stationary for the estimate of $\hl$, it is, however, possible to tune our strategy to adapt to nonstationary scenarios. In the stationary case, suppose the $n-th$ step reward is $R$,
	we had \[
	\hl^{(n)}=\hl^{(n-1)}+\frac{1}{n}\cdot\big(R-\hl^{(n-1)}\big).
	\]
	Now we rewrite the update in a more general form:
	\begin{equation}
		\hl^{(n)}=\hl^{(n-1)}+\alpha(n)\cdot\big(R-\hl^{(n-1)}\big).
	\end{equation} 
	Each time the weight for update would be $\alpha(n)$. $\alpha(n)=1/n$ corresponds to averaging the cumulative sum. We can take
	\[\alpha(n)=\frac{1}{\sum_{i=0}^{n-1} a^i},\, 0<a<1,
	\]
	for \emph{exponential weighted average}, which gives more weight in recent observations.
	A well known result to ensure the estimate to converge almost surely is ( \cite{Sutton1998} p. 33)
	\begin{equation}
		\sum_{i=1}^\infty \alpha(n)=\infty,\quad \sum_{i=1}^\infty \alpha^2(n)<\infty.
	\end{equation}
	\textbf{The first condition ensures the process to overcome the initial conditions or random fluctuations, the second condition assures convergence.}

The detailed implementation will be carried out in the Appendix.

\subsection{Rationality}
 Now instead of all possible strategies a player can play, let's confine ourselves to a specific class of strategies. 
 \begin{defn}[Rational Strategy]
 A \emph{rational strategy} is a function 
 \[ K: [0,1]\times \mathbb{R}\rightarrow [0,1]
 \]
 such that
 \begin{itemize}
 	\item $K$ is non-decreasing in $s$  $K(t,1)=1,\, K(t,s)=0,\, \forall t\in [0,1], \forall s<0$.
 	\item $K(t,s)=0$ for $s<t$. 
 \end{itemize}
 \end{defn}
\textbf{The first condition says that $K(t,\cdot)$ is a cumulative distribution function;
 the second condition says that thresholds less than the maximum of the previous valid scores $t$ will not be used.}
 \begin{remark}
 Mixed strategy of \emph{rational}  strategies is still \emph{rational}. We make the convention that the strategy for the first player is always \emph{rational}.
 \end{remark}
\begin{defn}[Rational Game]
We call a player is \emph{in a rational game}, if the player play first and the other players' strategy $\K=(K_1,K_2,\cdots,K_n)$ satisfies
\begin{itemize}
	\item $K_i$ is \emph{rational}, for $1\leq i\leq n$.
	\item $K_n=I_{\{s\geq t\}}$, where $I$ is indicator function.
\end{itemize}
\end{defn}
\begin{remark}
The second condition says that the last player will make the obvious choice: $t$ as threshold.
\end{remark}
From now on, we will assume that player of our interest is in a \emph{rational game}. Let's examine $L(t,K)$ more closely under this assumption. For $0<t<1,$ we have 
\begin{align*}
L(t,K)&=\int_0^1 G(t,s)dK(t)(s)\\
&=G(t)(s)K(t)(s)\Big|_0^1-\int_0^1 K(t,s)dG(t)(s)\\
&=1-\int_t^1 K(t,s)dG(t)(s)\\
&=1-\int_t^1 K(t,s)se^sds\\
\end{align*}
In particular, we have $L(t,K)\geq 1-F(t,1)$. 
\begin{theorem}\label{thm:general upper bound}
Suppose a player is in a \textbf{rational game} with $n+1$ players, then the optimal threshold for the  player is bounded above by \textbf{rational upper bound} $\gamma_n$, which is defined by
\begin{equation}
\int_{\gamma_n}^1[1-F(t,1)]dt=\big[1-F(\gamma_n,1)\big]^n
\end{equation} 
\end{theorem}
\begin{proof}
Since we have 
\[ 1-F(t,1)\leq L(t)\leq 1
\]
Therefore,
\begin{equation}\label{eq: ineq}
\frac{L(t)}{L(A)}\leq \frac{1}{1-F(A,1)}
\end{equation}
where we've suppressed $K$ in $L$.
Meanwhile,
\[\frac{M(t,\K)}{M(A,\K)}=\prod_{i=1}^{n-1} \frac{L(t,K_i)}{L(A,K_i)}\cdot \frac{G(t,t)}{G(A,A)}\leq \frac{1-F(t,1)}{[1-F(A,1)]^n}
\]
Now by equation~\ref{eq: general NE},we have
\begin{equation}
1=\int_A^1\frac{M(t,\K)}{M(A,\K)}dt \leq \frac{1}{[1-F(A,1)]^n}\int_A^1[1-F(t,1)]dt
\end{equation}
Since $\frac{1}{[1-F(A,1)]^n}\int_A^1[1-F(t,1)]dt$ is decreasing as a function of $A$, we have
\[ A\leq \gamma_n.
\]
\end{proof}

\begin{table}[ht]
	\caption{\emph{Rational Upper Bounds} $\gamma_n$} 
	\centering 
	\begin{tabular}{c c c c c c c c c} 
		\hline\hline
		n & 1&2&3&4&5&6& 7  \\ [0.5ex] 
		\hline 
		
		$\gamma_n$& 0.570557& 0.726417& 0.791326& 0.828415& 0.852904& 0.870488& 0.883829  \\ [0.7ex] 
		\hline 
		n &8&9&10&11&12&13&14  \\ [0.5ex] 
		\hline 
		
		$\gamma_n$& 0.894355&0.902905& 0.910009& 0.916021& 0.921184& 0.925674& 0.929619& \\[0.7ex]
		\hline\hline 
	\end{tabular}
	\label{table:rub} 
\end{table}
\begin{remark}\label{rmk: conterxp}
The \emph{rational upper bounds} $\gamma_n$ is  tight. In fact, the upper bound $\gamma_n$ is obtained when every other player's strategy is defined by
\begin{equation*}
K(t,s)= I_{\{ t\leq \gamma_n,t\leq s\}}+I_{\{t>\gamma_n,s=1\}}
\end{equation*}
where $I$ is indicator function.
\end{remark}
\begin{remark}
This upper bound estimation can help to optimize Model Free Strategy and Reinforcement Learning Strategy, see Appendix.
\end{remark}
\begin{figure}\label{fig:cmp}
	\centering
	\includegraphics[scale=0.27]{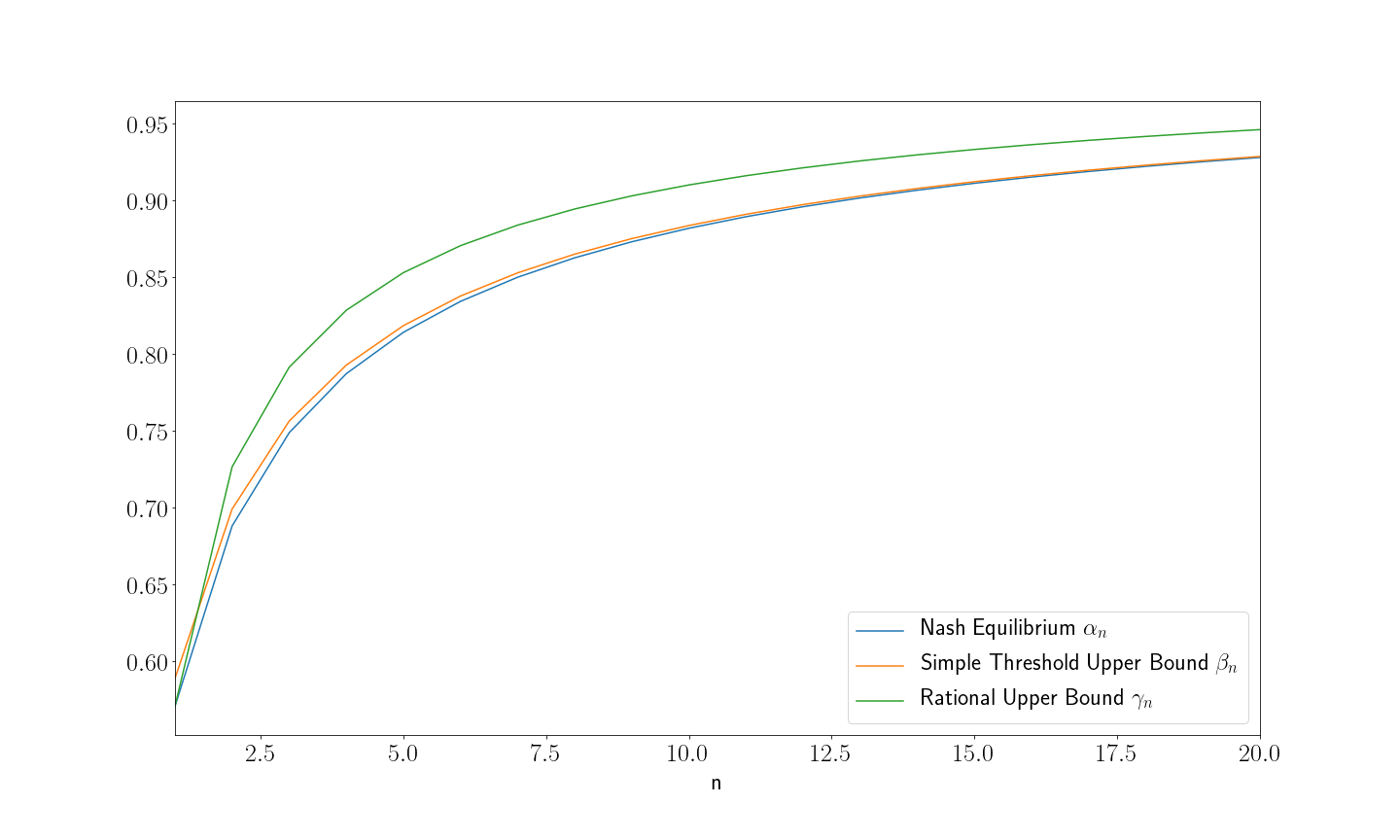}
	\caption{Comparison of $\alpha_n,\beta_n,\gamma_n$.}
\end{figure}
\subsection{Stability}
The strategy in remark~\ref{rmk: conterxp} is an example where player's response changes violently  according to the constraint $t$: when $t$ is below  level $\gamma_n$, the player choose threshold $t$; once $t$ exceeds $\gamma_n$, player's choice is suicidal. This behavior will drive up the level of the best response for the first player who want to exploit the situation by push other players over the cliff. This also explains the huge gap between the \emph{rational upper bound} $\gamma_n$ and the Nash Equilibrium $\alpha_n$, see Figure~\ref{fig:cmp}.

To exclude such behavior, we need to impose stability conditions on the strategies in our consideration as follows:
\begin{itemize}
	\item $K(t,s)$ is decreasing in $t$.
	\item The expected threshold of response
	\begin{equation}
	R(t,K):=\int_0^1sdK(t,s)=1-\int_0^1K(t,s)ds=1-\int_t^1K(t,s)ds
	\end{equation}
	is Lipschitz in $t$ under $L_1$ norm, i.e,
	\begin{equation}
	||R(t_1,K)-R(t_2,K)||_{L_2}\leq c(K)|t_1-t_2|
	\end{equation}
\end{itemize}

\textbf{
The first condition says that players
	won't be unnecessarily ``aggressive": faced with smaller constraint $t$, it is more likely for the player to chose smaller thresholds; The second condition says that the strategy's response will not change violently.
}

Let's see what these two conditions entail.
We will suppress $K$ whenever it causes no confusion. For $t_1<t_2$, we have
\[
L(t_2)-L(t_1)=\int_0^1\big[K(t_1,s)-K(t_2,s)\big]se^sds\geq 0\]
$L(t)$ is increasing, no surprise.

 Meanwhile, we have
\begin{align*}
L(t_2)-L(t_1)&=\int_0^1\big[K(t_1,s)-K(t_2,s)\big]se^sds\\ 
&\leq e\int_0^1\big[K(t_1,s)-K(t_2,s)\big]ds\\
&\leq e \cdot ||R(t_2)-R(t_1)||_{L_1}\\
&\leq ec(K)\cdot|t_2-t_1|
\end{align*}
$L(t)$ is also Lipschitz with a larger constant(by a factor of $e$)! This makes sense because we imposed stability condition on other players' strategies.

 Now we can see that these conditions are more handy if they are imposed on $L$ instead of $K$.
\begin{defn}[Stability]
 A rational strategy $K$ is $c$-\emph{stable} if
 \begin{itemize}
 	\item $L(t,K)$ is non-decreasing,
 	\item $|L(x,K)-L(y,K)|\leq c|x-y|$.
 \end{itemize}
\end{defn}
\begin{remark}
The constant $c$ regulates how stable the strategy is.
\end{remark}
\begin{remark}
Apparently, 
\emph{stability} is preserved under linear interpolation.
\end{remark}
\begin{defn}[Stable Game]
We say a player is in a $c$-\emph{stable game} if the player plays first and the other players' strategies are determined by $\K=(K_1,K_2,\cdots,K_n)$ such that
\begin{itemize}
	\item $K_i$ is $c$-\emph{stable}, for $1\leq i\leq n$.
	\item $K_n=I_{\{s\geq t\}}$, where $I$ is indicator function.
\end{itemize}
\end{defn}
As a consequence, equation~\ref{eq: general NE} have a unique solution $A^*$, and the expected payoff $E(A,\bl)$ is increasing for $A\leq A^*$ and decreasing afterwards. 
Therefore, we have the following proposition:
\begin{prop}\label{prop: simple threshold}
	In a $c$-\textbf{stable game} with constraint $t$, suppose $A$ is the optimal threshold without constraint(i.e, solution for equation~\ref{eq: general NE}),
	the optimal strategy is to take the larger of the two as the threshold.
\end{prop}

\begin{theorem}
Suppose a player is in a $c$-\textbf{stable game}  with $n+1$ players, then the optimal response of the first player is bounded above by \textbf{stable upper bound} $\theta_n(c)$, which is defined by
\begin{equation}
\int_{\theta_n}^1\min\big(1,F(\theta_n,1)+c(t-\theta_n)\big)^{n-1}[1-F(t,1)]dt=[1-F(\theta_n,1)]^n
\end{equation}
\end{theorem} 
\begin{proof}
Notice that\[
L(t)\leq \min(1,L(A)+c(t-A))
\]
Therefore,
\begin{align*}
\frac{L(t)}{L(A)}&\leq \min\Big(\frac{1}{L(A)},1+\frac{c(t-A)}{L(A)}\Big)\\
&\leq \min\Big(\frac{1}{1-F(A,1)},1+\frac{c(t-A)}{1-F(A,1)}\Big)
\end{align*}
The conclusion follows exactly the same argument in theorem~\ref{thm:general upper bound}.
\end{proof}
\subsection{Adaptive Threshold}
In this section we'll focus on an explicit subclass of adaptive strategies. We call a strategy $K$ is an \emph{adaptive threshold strategy} if 
\begin{equation}
K(t,s)=I_{\{s\geq\max(t,a)\}} \text{ for some } a.
\end{equation}
That is, the strategy will take the larger of $t$ and $a$ as threshold. This type of strategy corresponds the players who understand the structure of the Nash Equilibrium strategy but not knowing the exact values. Therefore, this type of strategy behaves less ``agressive" than Nash Equilibrium strategy: A samller $a$ does not make a difference while a larger $a$ makes life easier for the first player. Indeed, we have the following statement:
\begin{theorem}
Suppose $n+1$ players are in the game and everyone is playing \textbf{adaptive threshold strategy} except for possibly the first player. Then the best response of the first player is bounded above by the \textbf{Nash Equilibrium} $\alpha_n$.
\end{theorem}
\begin{proof}
Assume apart from the first player,
the other players' strategies are \[K_i:=I_{\{s\geq\max(t,a_i)\}},\]
 then by \refeq{eq: general NE}, the optimal threshold $A$ for the first player satisfies
\begin{equation}
\int_A^1\prod_{i=1}^n \big[1-F(\max(t,a_i),1)\big]dt=\prod_{i=1}^n \big[1-F(\max(A,a_i),1)\big].
\end{equation}
That is,
\begin{align*}
1&=\int_A^1\prod_{i=1}^n \frac{1-F(\max(t,a_i),1)}{1-F(\max(A,a_i),1)}dt\\
&=\int_A^1\prod_{\{a_i\leq A\}}\frac{1-F(\max(t,a_i),1)}{1-F(\max(A,a_i),1)}\cdot 
\prod_{\{a_i > A\}}\frac{1-F(\max(t,a_i),1)}{1-F(\max(A,a_i),1)} dt\\
&=\int_A^1  \bigg[\frac{1-F(t,1) }{ 1-F(A,1) }\bigg]^k \cdot \prod_{a_i > A}\frac{1-F(\max(t,a_i),1)}{1-F(a_i,1)} dt
\end{align*}
We notice that for $ a_i > A $,
\[ 1 \leq \frac{1-F(\max(t,a_i),1)}{1-F(a_i,1)} \leq \frac{1-F(t,1)}{1-F(A,1)} \]
and the above equation yields:
\[ \int_A^1 \bigg[\frac{1-F(t,1) }{ 1-F(A,1) }\bigg]^k dt
\leq \int_A^1\prod_{i=1}^n \frac{1-F(\max(t,a_i),1)}{1-F(\max(A,a_i),1)}dt 
\leq \int_A^1 \bigg[\frac{1-F(t,1) }{ 1-F(A,1) }\bigg]^n dt\]
By equation~\ref{eq:def NE}, we have $\alpha_k \leq  A\leq \alpha_n$, since
\[ \int_A^1\prod_{i=1}^n \frac{1-F(\max(t,a_i),1)}{1-F(\max(A,a_i),1)}dt\] is decreasing in $A$.
\end{proof}
\begin{remark}
The above theorem tells us that we should not play thresholds if we know other players are playing \emph{adaptive threshold} strategy. And we should play smaller threshold if at least one of the other players has a threshold greater than $\alpha_n$.

\end{remark}

\section{The Road to Reinforcement Learning Strategy}\label{sec:RL}
Now, let's have a closer look at the equation~\ref{eq: general payoff}:
\[
E(A,\K)=e^A\int_A^1 M(t,\K)dt.\]
Wouldn't it be great if we know the expected payoff? We'll have no problem in making our decision. Well, we couldn't know the expected payoff exactly, we can always estimate it. While estimate of $E(A,\K)$ for all $A$ is infeasible, we can always discretize the function, by the following fact:
\begin{equation}
\frac{\p E}{\p A}=e^A\bigg[\int_A^1 M(t,\K)dt-M(A,\K)\bigg]
\end{equation}
which is,
\[ -e^A\cdot M(A,\K)\leq \frac{\p E}{\p A} \leq e^A\int_A^1 M(t,\K)dt.
\]
Therefore, we have bounds for the derivative:
\[ -e\leq  \frac{\p E}{\p A} \leq e^A(1-A) \leq 1,
\]
i.e,
\[\bigg| \frac{\p E}{\p A} \bigg|\leq e.
\]

Let $\Delta:a=x_1<x_2<x_3\cdots x_n=b$ be a discretization of $[a,b]$, and $||\Delta||:=\max\{|x_i-x_{i-1}|\}$, we define 
\begin{equation}
E^\Delta(A,\K):=\frac{1}{x_{i+1}-x_i}\int_{x_i}^{x_{i+1}}E(t,\K)dt,\quad \forall t\in [x_i,x_{i+1}).
\end{equation}
We have 
\begin{equation}
	|E^\Delta(A,\K)-E(A,\K)|\leq \frac{e}{2}\cdot|x_{i+1}-x_i|\leq \frac{e}{2}||\Delta||.
\end{equation}
The first inequality is a consequence of lemma~\ref{lemma:good one}.
\begin{remark}
The accuracy is at the order of $O(1/m)$ if we take $||\Delta||=1/m$.
\end{remark} 
\subsection{Contextual Bandits}
In view of such analysis, we can just discretize the policy space $[0,1]$ into $m$ buckets, we don't know the expected payoff of each bucket, but we can estimate them. This is the $m$-armed bandit problem. There are lots of literatures on this topic. For example, chapter 2 of \cite{Sutton1998} gives an excellent introduction to this problem as well as its solutions. The following content will be dedicated to solve  the engineering challenges specific to this task:

\begin{itemize}
	\item For each permutation we need to assign an $m$-armed bandit, the total number is astronomical.
	\item For each $m$-armed bandit, if $m$ is small, the performance is capped by the large step size; if $m$ is large, then we need to do lots of exploration before we get reasonably good estimate of the reward for each arm.
	\item For $\epsilon$-greedy algorithm, the performance is also capped by the exploration rate $\epsilon$.
\end{itemize}

\subsubsection{Solution to the first issue}
Obviously it's too much for us to assign armed bandit for every permutation. If we make the assumption that the armed bandit only depends on who's \textbf{still} in the game, but not on their order, then we reduce the number of $m$-armed bandit to the order of $2^n$, which is still enormous for moderate size $n$.

By proposition~\ref{prop: simple threshold}, we know that the threshold we shall have is $\max(t,A)$, where $A$ is the output of our algorithm. If we are in a relatively late position in the game, it is \emph{highly} likely that at least one player before us hasn't gone bust, in which case, the constraint $t$ is usually larger than our policy $A$, therefore \textbf{the accuracy of $A$ has minimal effect on the decision}. Indeed, let's look at the case when everyone is playing Nash Equilibrium. Let $A_k:=\alpha_k+\delta$, then the Second derivative at $\alpha_k$ is (see figure~\ref{fig:de})
\begin{equation}
	\Delta_k= - e^{\alpha_k}\prod_{i=k+1}^{n}[1-F(\alpha_i,1)]\cdot [(1-F(\alpha_k))^k+k(1-F(\alpha_k))^{k-1}].
\end{equation}
Therefore, the derivative at $A_k$ is approximately $\Delta_k\cdot \delta$. That is, the effect a small change near $A_k$ would have proportional impact on the payoff with a factor of $\Delta_k\cdot\delta$. 
\begin{figure}\label{fig:de}
	\centering
	\includegraphics[scale=0.25]{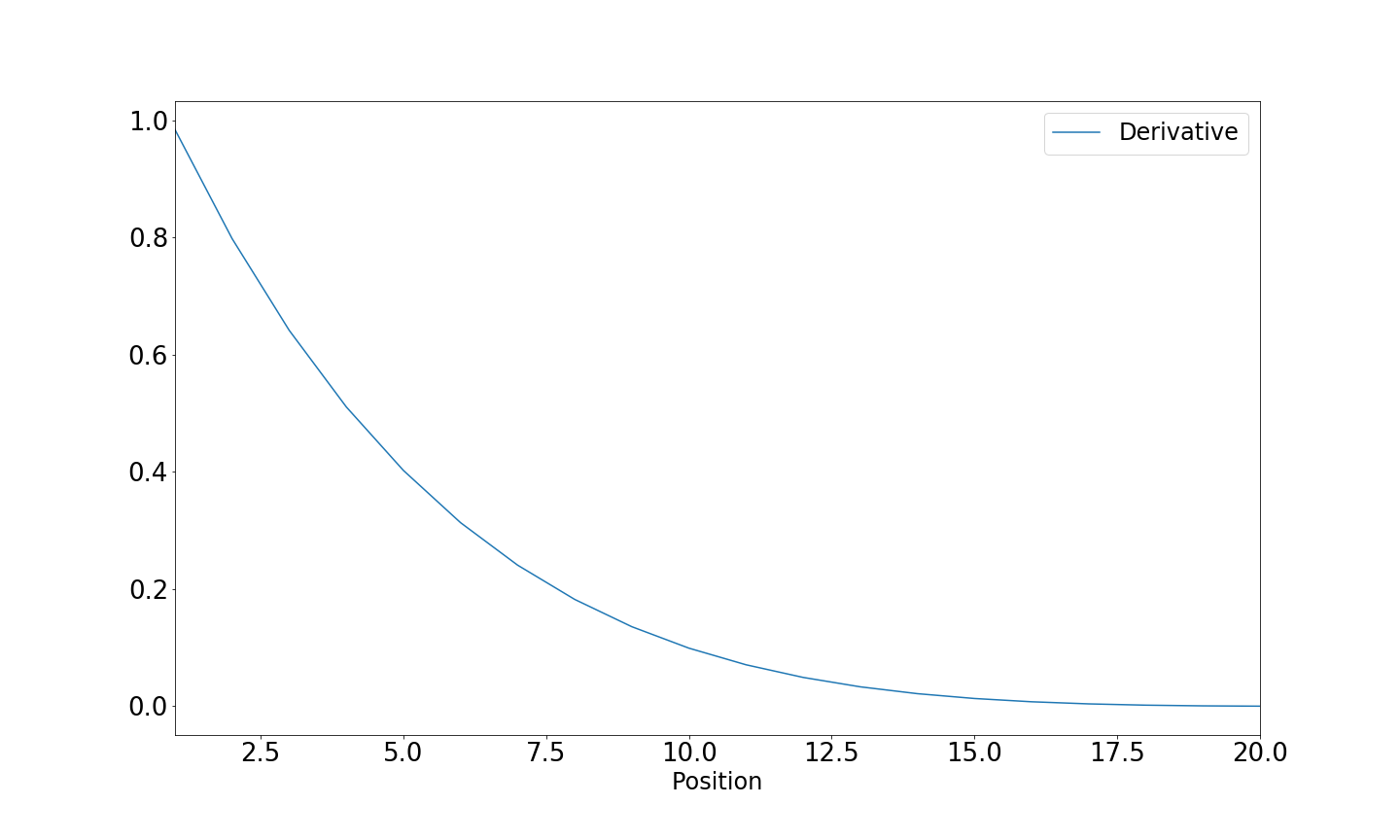}
	\caption{Second Derivative for game with $20$ players.}
\end{figure}
Furthermore, only a small fraction of the rounds played actually contributes to the estimation of rewards of arms with smaller thresholds for late positions. That means the estimations for rewards in later positions have significantly larger error. This also points to the direction that the effort to increase accuracy for late game threshold is not as cost-effective as that for early game.

In view of this analysis, we can reduce the number of $m$-armed bandits considerably. Depending on the resource available to us, we can choose to assign $m$-armed bandits when we are in relatively early position in the game. For example, we can assign $1$, $n$ and $\frac{n(n-1)}{2}$ armed bandits if we are in the first, second, third place in the game, respectively; for each later position, we only assign one $m$-armed bandit. In this case, we only need $O(n^2)$ bandits.
\subsubsection{Solution to second issue}
By proposition~\ref{thm:general upper bound}, we know that not all thresholds are valid candidates. This helps a little bit. What is likely to happen is that lots of arms are barely pulled since they are clearly not good choices. Therefore what we can do is to use \textbf{unequal-step-size discretization}.

\textbf{Policy Pruning}: In order to facilitate the training, we can prune and shrink  the policy space as the game goes: after certain number of rounds, we dump the bottom, say,  $10\%$ performance arms which have no real chance to be the optimal policy. Then we split each of the top  $10\%$ into two arms to gain accuracy. And the total number of arms remains the same. Notice that this process will shrink the ``policy space". In the case where others' strategies are not stationary, we shall not shrink the ``policy space" too much. Usually, we can predetermined how many times we want to use this process. It is also worth mentioning that \textbf{ the further in the game, the smaller difference between the rewards of different arms. To overcome the random fluctuations, it's advisable to take longer time to perform the next policy pruning process. To rank the expected rewards of two Bernoulli process with difference $\delta$ at certain confidence level, it requires $O(\frac{1}{\delta^2})$ trials.}
\subsubsection{Solution to the third issue}
If we do not change the size of discretization, we can simply reduce the exploration rate every given rounds. If we shrink the policy space, we don't have to reduce the exploration rate $\epsilon$ till the process terminates. One way to avoid $\epsilon$ altogether is to use \textbf{gradient bandit algorithm} or \textbf{UCB algorithm}.

\section{Appendix}
\subsection{Bench Mark}
We'll need a good metric to measure the performance of different algorithms, preferiably one does not depend on the number of players in a game. We propose two different metrics:
\subsubsection{Scale Metric}
For a game with $n$ players, the \emph{scale metric} $s$ for $K$ is defined as 
\[ s(K):=nR(K)
\] 
where $R$ is the reward. The player get average score will have $s=1$. This metric can be applied anywhere, the downside is the performance is highly dependent on the opponents. Therefore, the metric can not be interprated out of context. 
\subsubsection{Reference Metric}For a game with $n$ players, the \emph{reference metric} $r$ for $K$ is defined as
\[ r(K):=\frac{R(K)-R(K_0)}{R(K_{ne})-R(K_0)}
\]
Where $K_{ne}$ is the Nash Equilibrium strategy, and $K_0$ is a reference strategy. For example, 
\[ K_0(t,s)=\max\Big(0,\frac{s-t}{1-t}\Big).
\]
This reference is quite useful for the exploration-exploitation type algorithms, because the exploration part is exactly this $K_0$. Suppose the exploration rate is $\epsilon$ , the exploitation part is denoted by $K_{T}$( since it's time dependent) and the average payoff of $K_{T}$ up to time $T$ is denoted by $\bar{R}_T$, then
\[ r(K_T)=(1-\epsilon)\cdot\frac{\bar{R}_T-R(K_0)}{R_{ne}-R(K_0)}.
\]
Therefore, if $r(K_T)\geq 1-\epsilon$, then this strategy can outperform the Nash Equilibrium over the long run as we decrease the exploration rate over time. The disadvantege of such metric is the requirement of the presence of both $K_{ne}$ and $K_0$ in the game.

\subsection{Algorithm}
\begin{algorithm}
	\caption{Model Free Strategy}\label{alg:update profile}
	\begin{algorithmic}
		\Procedure{ModelfreeStrategy} {$N,n$}\Comment{This game has $N$ rounds with $n$ players}
		\For{$i:=2\to n$} \Comment{Initialize profiles for player $2\cdots,n$.}
		\For{$j:=0 \to n$}\Comment{Initialize profiles for each position.}
		\State Initialize $\fp_i[j](t)\gets 0,\,t\in[0,1]$. \Comment{$\fp_i[j](t)$ is a function.}
		\EndFor
		\EndFor
		\State Initialize $r\gets1$ \Comment{$r$ is the number of rounds played so far.}
		\While {$r\leq N$}
		\State {PList $\gets$ reshuffle($1,2,\cdots,n$)}\Comment{ PList is a permutation of $1,2,\cdots,n$.}
		\State Initialize $t\gets0$
		\For{$i:=1 \to n$} \Comment{ Players play game in order.}
		\State  $p\gets$ PList[$i$]\Comment{$p$ is players' Id.}
		\If{ $p=1$}\Comment{We are player $1$ }
		\State Initialize $M\gets1$ \Comment{$M$ is the constant function of $1$.}
		\For{$j:=i+1 \to n$}
		\State $k\gets$ PList[$j$]\Comment{Player $k$ is in position $j$.}
		\State $M\gets M*\fp_k[j](t)$
		\EndFor
		\State $A\gets \arg\max\limits_{ s\geq t} e^s\int_s^1 M(s)ds$ \Comment{$A$ is the optimal threshold.}
		\State  $w \gets$ Play game with threshold $A$ \Comment{$w$ is the score.}
		\Else
		\State $w\gets$ Player $p$ Play Game \Comment{Player $p$ play game.}
		\State $R\gets 0$\Comment{Initialize reward.}
		\If{$w< t$}: 
		\State $R\gets 1$ \Comment{Update reward.}
		\EndIf
		\State
		$\fp_p[i](t)\gets\fp_p[i](t)+\alpha(r)(R-\fp_p[i](t))$\Comment{$\alpha(r)$ is step weight.}
		\EndIf
		\State $t=\max(t,w)$ \Comment{Update the maximum of scores.}
		\EndFor
		\State $r\gets r+1$
		\EndWhile 
		\State \textbf{return}
		\EndProcedure
		\end{algorithmic}
\end{algorithm}

\begin{algorithm}
	\caption{Reinforcement Learning Strategy}
	\begin{algorithmic}
		\Procedure{$\epsilon$-Greedy Algorithm} {$n,\epsilon,T$ } \Comment{$T$ is timetable for Policy pruning.}
		\State Initialize $r_{i,j}\gets1,N_i\gets 10$\Comment{High initial condition encourages exploration.}
		\Loop 
		\State $ R,(i,j)\gets$ Play game with exploration rate $\epsilon$. \Comment{$R$: reward.}
		
		\State $N_i\gets N_i+1$.
		\State $r_{i,j}\gets r_{i,j}+\frac{1}{N_i}(R-r_{i,j})$
		\State $r_i,N_i\gets T(N_i)$\Comment{Policy pruning according to timetable $T$.}
	
		\EndLoop
		\EndProcedure
	\end{algorithmic}
\end{algorithm}
 
\bibliographystyle{alpha}
\bibliography{ref}
\end{document}